\documentclass[10pt]{article} 
\usepackage{amsthm}
\theoremstyle{definition}
\newtheorem{definition}{Definition}[section]

\theoremstyle{plain}
\newtheorem{theorem}[definition]{Theorem}

\newtheorem{lemma}[definition]{Lemma}

\newtheorem{assumption}[definition]{Assumption}
\usepackage[english]{babel}
\usepackage{indentfirst}
\usepackage{fancyhdr}
\usepackage{graphicx}
\usepackage{newlfont}
\usepackage{amssymb} 
\usepackage{amsmath}
\usepackage{latexsym}
\usepackage{amsthm}
\usepackage{eucal}
\usepackage{eufrak}
\usepackage{bbold}
 \usepackage{tocbibind}
 \usepackage{newlfont}
 \usepackage{enumerate}
 \usepackage{listings}
 \usepackage{tikz}
 \usepackage{bbm}
\usepackage{color}
\usepackage[toc,page]{appendix}
\usepackage{hyperref}
\usepackage{leftidx}
\allowdisplaybreaks

\numberwithin{equation}{section}

\usepackage{hyperref}
\usepackage{leftidx}

\usepackage{color}

\newcommand{\blue}[1]{\textcolor{black}{#1}}

\usepackage[square,sort&compress,numbers]{natbib}

\usepackage{color} 

\numberwithin{equation}{section}

\definecolor{mygreen}{RGB}{28,172,0} 
\definecolor{mylilas}{RGB}{170,55,241}

   \newcommand*\samethanks[1][\value{footnote}]{\footnotemark[#1]}

\newcommand{\RN}[1]{
  \textup{\uppercase\expandafter{\romannumeral#1}}
} 

\author{Francesca Biagini \thanks{Workgroup Financial Mathematics, Department of Mathematics, Ludwig-Maximilians-Universit\"{a}t M\"{u}nchen, Theresienstr. 39, 80333 Munich, Germany.  Emails: biagini@math.lmu.de,  meyer-brandis@math.lmu.de, oberpriller@math.lmu.de.} \and  Andrea Mazzon\thanks{Department of Economics, University of Verona, Via Cantarane 24, 37129 Verona, Italy. E-mail: andrea.mazzon@univr.it} \and Thilo Meyer-Brandis\samethanks[1] \and  Katharina Oberpriller\samethanks[1]   }
 
\title{Supplement to ``Liquidity based modeling of asset price bubbles via random matching''}

\begin{document}

\maketitle

This is a supplement to the paper  \cite{biagini_mazzon_oberpriller_meyer_brandis_2022}. The supplement is organized as follows. First, we prove {Theorem 3.13 in \cite{biagini_mazzon_oberpriller_meyer_brandis_2022}} which provides the existence of the dynamical system $\mathbb{D}$ introduced in {Definition 3.6 in \cite{biagini_mazzon_oberpriller_meyer_brandis_2022}}. Second, we show some properties of $\mathbb{D}$ which are summarized in {Theorem 3.14 in \cite{biagini_mazzon_oberpriller_meyer_brandis_2022}}.

In the following, we only state the basic setting and refer to \cite{biagini_mazzon_oberpriller_meyer_brandis_2022} for definitions.

\section{Setting}

 Let $(\tilde \Omega, \tilde{\mathcal{F}}, \tilde P)$ be a probability space and $(\hat \Omega, \hat{\mathcal{F}})$ another measurable space. We define the product space  
\begin{equation} \label{eq:ProductSpace}
	(\Omega,  \mathcal{F}):={(\tilde \Omega \times \hat \Omega , \tilde{\mathcal{F}} \otimes \hat{\mathcal{F}})}.
\end{equation} 
Let $\hat{P}$ be a Markov kernel (or stochastic kernel) from $\tilde \Omega$ to $\hat \Omega$. Given $\tilde \omega \in \tilde \Omega,$ we set $\hat P^{\tilde \omega}:=\hat P(\tilde \omega)$ with a slight notational abuse. We then introduce a probability measure $P$ on $(\Omega,  \mathcal{F})$ as the semidirect product of $\tilde P$ and $\hat P$, that is,
\begin{equation}\label{eq:ProbhatPSetting}
P({\tilde A \times \hat{A}}) := (\tilde P \ltimes \hat P)({\tilde A \times \hat{A}}) =  \int_{\tilde A}\hat P^{\tilde \omega}(\hat A)d\tilde P(\tilde \omega).
\end{equation} 
We fix an atomless probability space $(I,\cal{I},\lambda)$ representing the space of agents and let $(I \times \Omega, \cal{I} \boxtimes \cal{F},\lambda \boxtimes P)$ be {a rich} Fubini extension of $(I \times \Omega, \cal{I} \otimes \cal{F}, \lambda \otimes P)$. All agents in ${I}$ can be classified according to their type. In particular, we let $S=\lbrace 1,2,...,K \rbrace$ be a finite space of types and  say that an agent has type $J$ if he is not matched. We denote by $\hat{S}:=S \times (S \cup \lbrace J \rbrace)$ the extended type space. Moreover, we call $\hat{\Delta}$ the space of extended type distributions, which is the set of probability distributions ${p}$ on $\hat{S}$ satisfying ${p}(k,l)={p}(l,k)$ for any $k$ and $l$ in $S$. This space is endowed with the topology $\cal{T}^{\Delta}$ induced by the topology of the space of matrices with $|S|$ rows and $|S| + 1$ columns. We consider $(n)_{n \geq 1}$ time periods and denote by $(\eta^n, \theta^n, \xi^n, \sigma^n, \varsigma^n)$ the matrix valued processes, with $(\eta^n, \theta^n, \xi^n, \sigma^n, \varsigma^n)=(\eta^n_{kl},\theta^n_{kl}, \xi^n_{kl}, \sigma^n_{kl}[r,s], \varsigma^n_{kl}[r])_{k,l,r,s \in S \times S \times S \times S }$ for $n \geq 1$, on $(\Omega, \mathcal{F}, P)$. For a detailed introduction of these processes we refer to {Section 3 in \cite{biagini_mazzon_oberpriller_meyer_brandis_2022}}. Moreover, let $\hat{p}=(\hat{p}^n)_{n \geq 1}$ be a stochastic process on $(\Omega, \mathcal{F},P)$ with values in $\hat{\Delta}$, {representing the  evolution of the underlying extended type distribution}. We assume that $\hat{p}^0$ is deterministic. \\
Given the input processes $(\eta, \theta, \xi,\sigma, \varsigma)$ we denote by $\mathbb{D}$ a dynamical system on $(I \times \Omega, \cal{I} \boxtimes \cal{F},\lambda \boxtimes P)$ and by $\Pi=(\alpha, \pi,g)=(\alpha^n, \pi^n, g^n)_{n \in \mathbb{N} \backslash \lbrace{ 0 \rbrace}}$ the agent-type function, the random matching and the partner-type function, respectively, as introduced in {Definition 3.6 in \cite{biagini_mazzon_oberpriller_meyer_brandis_2022}}, which we recall in the following.

\begin{definition} \label{defi:DynamicalSystemDiscrete}
A \emph{dynamical system} $\mathbb{D}$ defined on $(I \times \Omega, \cal{I} \boxtimes \mathcal{F},\lambda \boxtimes P )$ is a triple $\Pi=(\alpha, \pi,g)=(\alpha^n, \pi^n, g^n)_{n \in \mathbb{N} \backslash \lbrace 0 \rbrace}$ such that for each integer period $n\geq 1$ we have:
\begin{enumerate}
	\item $\alpha^n: I \times {\Omega} \to S$ is the $\cal{I} \boxtimes {\mathcal{F}}$-measurable agent type function. The corresponding end-of-period type of agent $i$ under the realization ${ \omega} \in {\Omega}$ is given by $\alpha^n(i,{\omega}) \in S$.
	\item A random matching $\pi^n: I \times {\Omega} \to I$, describing the end-of-period agent $\pi^n(i)$ to whom agent $i$ is currently matched, if agent $i$ is currently matched. If agent $i$ is not matched, then $\pi^{n}(i)=i$. The associated $\cal{I} \boxtimes {\mathcal{F}}$-measurable partner-type function $g^{n}:I \times {\Omega} \to S \cup \lbrace J \rbrace$ is given by
		\[g^n(i,{\omega})=\begin{cases}
	\alpha^n(\pi^n(i,{\omega}),{\omega}) & \text{ if } \pi^n(i,{\omega}) \neq i\\
		J& \text{ if } \pi^n(i,{\omega}) = i,
	\end{cases}\]
	providing the type of the agent to whom agent $i$ is matched, if agent $i$ is matched, or $J$ if agent $i$ is not matched. 
\end{enumerate}
Let the initial condition $\Pi^0=(\alpha^0,\beta^0)$ of $\mathbb{D}$ be given.  
We now construct a dynamical system $\mathbb{D}$ defined on $(I \times \Omega, \cal{I} \boxtimes \mathcal{F},\lambda \boxtimes P )$ with {input processes} $(\eta^n,\theta^n,\xi^n, \sigma^n, \varsigma^n)_{n\geq 1}$. We assume that $\Pi^{n-1}=(\alpha^{n-1},\pi^{n-1},g^{n-1})$ is given for some $n \geq 1$, and define $\Pi^{n}=(\alpha^n,\pi^n, g^n)$ by characterizing the three sub-steps of random change of types of agents, random matchings, break-ups and possible type changes after matchings and break-ups as follows. \\
\textbf{Mutation:} For $n \geq 1$ consider an $\cal{I} \boxtimes {\mathcal{F}}$-measurable post mutation function $$\bar{\alpha}^n: I \times \Omega \to S.$$ In particular, $\bar{\alpha}_i^n({\omega}):=\bar{\alpha}^n(i,\omega)$ is the type of agent $i$ after the random mutation under the scenario ${\omega} \in {\Omega}$. The type of the agent to whom {an agent} is matched is identified by a $\cal{I} \boxtimes {\mathcal{F}}$-measurable function $$\bar{g}^n:I \times \Omega \to S {\cup \lbrace J \rbrace},$$ given by 
$$
\bar{g}^n(i,{\omega})=\bar{\alpha}^n(\pi^{n-1}(i,{\omega}),{\omega})
$$ 
for any ${\omega} \in {\Omega}$. {In particular, $\bar{g}_i^n({\omega}):=\bar{g}^n(i,\omega)$ is the type of the agent to whom {an agent} is matched  under the scenario ${\omega} \in {\Omega}$.} Given $\hat{p}^{n-1}$ and $\tilde \omega \in \tilde \Omega$, for any $k_1,k_2,l_1$ and $l_2$ in $S$, for any $r \in S \cup \lbrace J \rbrace$, for $\lambda$-almost every agent $i$, we set
\small{
\begin{align}
&\hat P^{\tilde \omega}\left(\bar{\alpha}_i^n{(\tilde{\omega},\cdot)}=k_2, \bar{g}_i^n{(\tilde{\omega},\cdot)}=l_2 \vert \alpha_{i}^{n-1}{(\tilde{\omega},\cdot)}=k_1, g_i^{n-1}{(\tilde{\omega},\cdot)}=l_1, \hat{p}^{n-1} {(\tilde{\omega},\cdot)}\right)\blue{(\hat{\omega})}\nonumber \\&\quad=\eta_{k_1,k_2}\left( \tilde \omega,n,\hat{p}^{n-1}(\tilde{\omega},{\hat{\omega}})\right)\eta_{l_1,l_2}\left( \tilde \omega,n,\hat{p}^{n-1}(\tilde{\omega},{\hat{\omega}})\right), \label{eq:IndiMutation1}
\end{align} 
\begin{align}
&\hat P^{\tilde \omega}\left(\bar{\alpha}_i^n{(\tilde{\omega},\cdot)}=k_2, \bar{g}_i^n{(\tilde{\omega},\cdot)}=r \vert \alpha_{i}^{n-1}{(\tilde{\omega},\cdot)}=k_1, g_i^{n-1}{(\tilde{\omega},\cdot)}=J, \hat{p}^{n-1}{(\tilde{\omega},\cdot)}\right)\blue{(\hat{\omega})} \nonumber \\
&\quad =\eta_{k_1,k_2}\left(\tilde{\omega},n,\hat{p}^{n-1}(\tilde{\omega},{\hat{\omega}})\right)\delta_J(r), \label{eq:IndiMutation2}
\end{align}}
We then set 
$$
\bar{\beta}^n{(\omega)}=(\bar{\alpha}^n{(\omega)}, \bar{g}^n{(\omega)}), \quad n \geq 1.
$$
The post-mutation extended type distribution realized in the state of the world ${\omega \in \Omega}$ is {denoted by $\check{p}(\omega)=(\check{p}^n(\omega)[k,l])_{k \in S,  l \in S \cup J}$, where 
	\begin{equation}\label{eq:PostMutationExtendedType}
	\check{p}^n(\omega)[k,l]:=\lambda(\lbrace i \in I: \bar{\alpha}^n(i,\omega)=k,  \bar{g}^n(i,\omega)=l\rbrace).
	\end{equation}
}
\textbf{Matching:} {We introduce a random matching} $\bar{\pi}^n: I \times { \Omega} \to I$ {and the associated post-matching partner type function $\bar{\bar{g}}^n$ given by
\[\bar{\bar{g}}^n(i,{ \omega})=\begin{cases}
\bar{\alpha}^n(\bar{\pi}^n(i,{ \omega}),{ \omega}) & \text{ if } \bar{\pi}^n(i,{ \omega}) \neq i\\
J& \text{ if }  \bar{\pi}^n(i,{ \omega})=i,
\end{cases}\]
satisfying} the following properties:
\begin{enumerate}
\item $\bar{\bar{g}}^n$ is $\cal{I} \boxtimes {\mathcal{F}}$-measurable.
	\item {For any $\tilde \omega \in \tilde \Omega$, any $k,l \in S$ and any $r \in S \cup \lbrace J \rbrace$, it holds
\begin{equation} \notag
	\hat P^{\tilde \omega}(\bar{\bar{g}}^n{(\tilde{\omega},\cdot)}=r \vert \bar{\alpha}^n_i{(\tilde{\omega},\cdot)}=k, \bar{g}_i^n{(\tilde{\omega},\cdot)}=l)(\hat{\omega})=\delta_l(r).
\end{equation}
This means that
	\begin{equation} \notag
		\bar{\pi}^{n}_{{ \omega}}(i)=\pi_{{ \omega}}^{n-1}(i) \quad \text{ for any } i \in \lbrace i: \pi^{n-1}(i,{ \omega}) \neq i \rbrace.
	\end{equation}}
\item Given $\tilde{\omega}\in \tilde{\Omega}$ and the post-mutation extended type distribution $\check{p}^n$ in \eqref{eq:PostMutationExtendedType}, an unmatched agent of type $k$ is matched to a unmatched agent of type $l$ with conditional probability {$\theta_{kl}(\tilde{\omega},n,\check{p}^n)$}, that is for $\lambda$-almost every agent $i$ and $\hat P^{\tilde{\omega}}$-almost every $\hat{\omega}$, we define
\begin{equation} \label{eq:MatchingCondProb1}
	\hat P^{\tilde \omega}(\bar{\bar{g}}^n{(\tilde{\omega},\cdot)}=l \vert \bar{\alpha}^n_i{(\tilde{\omega},\cdot)}=k, \bar{g}_i^n{(\tilde{\omega},\cdot)}=J, \check{p}^n{(\tilde{\omega},\cdot)})(\hat{\omega})=\theta_{kl}^n(\tilde \omega, \check{p}^{n}(\tilde{\omega},\hat{{\omega}})).
\end{equation}
This also implies that
\begin{equation} \label{eq:MatchingCondProb2}
	\hat P^{\tilde \omega}(\bar{\bar{g}}^n{(\tilde{\omega},\cdot)}=J \vert \bar{\alpha}^n_i{(\tilde{\omega},\cdot)}=k, \bar{g}_i^n{(\tilde{\omega},\cdot)}=J, \check{p}^n{(\tilde{\omega},\cdot)})(\hat{\omega})=1-\sum_{l \in S} \theta_{kl}^n(\tilde{\omega},\check{p}^n(\tilde{\omega},{\hat{\omega}}))=b^k(\tilde{\omega},\check{p}^n(\tilde{\omega},{\hat{\omega}})).
\end{equation}
\end{enumerate}
The extended type of agent $i$ after the random matching step is 
$$\bar{\bar{\beta}}^n_i{(\omega)}=(\bar{\alpha}_i^n{(\omega)},\bar{\bar{g}}_i^n{(\omega)}), \quad n \geq 1.$$ 
{We denote the post-matching extended type distribution realized in ${ \omega} \in {\Omega}$ by $\check{\check{p}}^n(\omega)=(\check{\check{p}}^n(\omega)[k,l])_{k \in S,  l \in S \cup J}$, where 
	\begin{equation}\label{eq:PostMatchingExtendedType}
	\check{\check{p}}^n(\omega)[k,l]:=\lambda(\lbrace i \in I: \bar{\bar{\alpha}}^n(i,\omega)=k,  \bar{g}^n(i,\omega)=l\rbrace).
	\end{equation}}

\textbf{Type changes of matched agents with break-up:} We now define a random matching $\pi^n$ {by}
\begin{equation}
\pi^n(i)=\begin{cases} \label{eq:BreakUpPi}
	\bar{\pi}^n(i) & \text{ if } \bar\pi^n(i) \neq i\\ 
		i& \text{ if } \bar\pi^n(i) = i.
	\end{cases}
\end{equation}
We then introduce an $(\cal{I} \boxtimes {\mathcal{F}})$-measurable agent type function $\alpha^n$ and an $(\cal{I} \boxtimes {\mathcal{F}})$-measurable partner function $g^n$ {with} 
$$
g^n(i,{\omega})=\alpha^n(\pi^n(i,{\omega}), {\omega}), \quad n \geq 1,
$$
for all $(i,{\omega}) \in I \times {\Omega}$. Given $\tilde{\omega} \in \tilde{\Omega}$, $\check{\check{p}}^n \in \hat{\Delta}$, for any $k_1,k_2,l_1,l_2 \in S$ and $r \in S \cup \lbrace J \rbrace$, for $\lambda$-almost every agent $i$, and for $\hat{P}^{\tilde{\omega}}$-almost every $\hat \omega$, we set 
\begin{align}
	\hat P^{\tilde \omega}\left(\alpha_i^n{(\tilde{\omega},\cdot)}=l_1, g_i^n(\tilde{\omega},\cdot)=r \vert \bar{\alpha}_i^n{(\tilde{\omega},\cdot)}=k_1, \bar{\bar{g}}^n_i{(\tilde{\omega},\cdot)}=J\right)\blue{(\hat{\omega})}=\delta_{k_1}(l_1) \delta_J(r), \label{eq:BreakUpCondProb0.1}
\end{align}
\begin{align}
	&\hat P^{\tilde \omega}\left(\alpha_i^n{(\tilde{\omega},\cdot)}=l_1, g_i^n{(\tilde{\omega},\cdot)}=l_2 \vert \bar{\alpha}_i^n{(\tilde{\omega},\cdot)}=k_1, \bar{\bar{g}}^n_i{(\tilde{\omega},\cdot)}=k_2, \check{\check{p}}^n{(\tilde{\omega},\cdot)} \right)\blue{(\hat{\omega})} \nonumber \\ &\quad =\left(1-\xi_{k_1k_2}( \tilde \omega, n, \check{\check{p}}^n(\tilde{\omega}, {\hat{\omega}}))\right) \sigma_{k_1 k_2}[l_1,l_2](\tilde \omega, n, \check{\check{p}}^n(\tilde{\omega},{\hat{\omega}})), \label{eq:BreakUpCondProb1} 
\end{align}
\begin{align}
	&\hat P^{\tilde \omega}\left(\alpha_i^n{(\tilde{\omega},\cdot)}=l_1, g_i^n{(\tilde{\omega},\cdot)}=J \vert \bar{\alpha}_i^n{(\tilde{\omega},\cdot)}=k_1, \bar{\bar{g}}^n_i{(\tilde{\omega},\cdot)}=k_2, \check{\check{p}}^n {(\tilde{\omega},\cdot)} \right)\blue{(\hat{\omega})} \nonumber \\& \quad=\xi_{k_1k_2}(\tilde \omega, n,\check{\check{p}}^n(\tilde{\omega},{\hat{\omega}})) \varsigma_{k_1 k_2}^n[l_1](\tilde \omega, n, \check{\check{p}}^n(\tilde{\omega},{\hat{\omega}})). \label{eq:BreakUpCondProb2}
	\end{align}
	The extended-type function at the end of the period is $$\beta^n{({\omega})}=(\alpha^n{({\omega})},g^n{({\omega})}), \quad n \geq 1.$$ 
\end{definition}	
{We denote the extended type distribution at the end of period $n$ realized in ${ \omega} \in {\Omega}$ by ${\hat{p}}^n(\omega)=({\hat{p}}^n(\omega)[k,l])_{k \in S,  l \in S \cup J}$, where 
	\begin{equation}\label{eq:ExtendedTypeEndOfPeriod}
	\hat{{p}}^n(\omega)[k,l]:=\lambda(\lbrace i \in I: {\alpha}^n(i,\omega)=k,  {g}^n(i,\omega)=l\rbrace).
	\end{equation}}
	
Furthermore, the definition of  Markov conditionally independent (MCI) dynamical system is provided in {Definition 3.8 in \cite{biagini_mazzon_oberpriller_meyer_brandis_2022}}. We work under the following assumption, which is {Assumption 3.9 in \cite{biagini_mazzon_oberpriller_meyer_brandis_2022}}.
\begin{assumption} \label{assum:ExistenceHyperfinite}
	Let $(\tilde{\Omega}, \tilde{\mathcal{F}}, \tilde{P})$ be the probability space introduced. We assume that there exists its corresponding hyperfinite internal probability space, which we denote from now on also by $(\tilde{\Omega}, \tilde{\mathcal{F}}, \tilde{P})$ by a slight notational abuse. 
\end{assumption}

As already pointed out in \cite{biagini_mazzon_oberpriller_meyer_brandis_2022}, the proofs of the results below follow by analogous arguments as in \cite{RandomMatchingDiscrete} which is possible due to the product structure of the space $\Omega$ in \eqref{eq:ProductSpace} and the Markov kernel $P$ in \eqref{eq:ProbhatPSetting}. As in \cite{RandomMatchingDiscrete} we use some concepts and notations from nonstandard analysis. Note here that an object with an upper left star means the transfer of a standard object to the nonstandard universe. For a detailed overview of the necessary tools of nonstandard analysis, we refer to Appendix D.2. in \cite{RandomMatchingDiscrete}.

\section{Proof of {Theorem 3.13 in \cite{biagini_mazzon_oberpriller_meyer_brandis_2022}}}

From now on, we fix the hyperfinite internal space $(\tilde \Omega, \tilde{\mathcal{F}}, \tilde P)$, along with the input functions \linebreak 
$(\eta_{kl},\theta_{kl}, \xi_{kl}, \sigma_{kl}[r,s], \varsigma_{kl}[r])_{k,l,r,s \in S \times S \times S \times S }$  from $\tilde{\Omega} \times \mathbb{N} \times \Delta$ to $[0,1]$ introduced above. Given this framework we prove the existence of a {rich} Fubini extension $(I \times \Omega, \cal{I} \boxtimes \mathcal{F}, \lambda \boxtimes P)$, on which a dynamical system $\mathbb{D}$ described in Definition \ref{defi:DynamicalSystemDiscrete} for such input probabilities is defined. More specifically, we are going to construct the space $\hat{\Omega}$ and the probability measure $\hat{P}$ such that  $\Omega = \tilde{\Omega} \times \hat{\Omega}$ and $P = \tilde P \ltimes \hat P$ is a Markov kernel from $\tilde \Omega$ to $\hat \Omega$. \\
We now present and prove {Theorem 3.13 in \cite{biagini_mazzon_oberpriller_meyer_brandis_2022}}. The proof is based on Proposition 3.12 in \cite{biagini_mazzon_oberpriller_meyer_brandis_2022}, which focuses on the random matching step and shows the existence of a suitable hyperfinite probability space and partial matching, generalizing Lemma 7 in \cite{RandomMatchingDiscrete}.

\begin{theorem} 
Let {Assumption 3.9 in \cite{biagini_mazzon_oberpriller_meyer_brandis_2022}} hold and {$(\eta_{kl},\theta_{kl}, \xi_{kl}, \sigma_{kl}[r,s], \varsigma_{kl}[r])_{k,l,r,s \in S \times S \times S \times S }$ be the input functions from $\tilde{\Omega} \times \mathbb{N} \times \hat\Delta$} defined in {Section 3 in \cite{biagini_mazzon_oberpriller_meyer_brandis_2022}}. Then for any extended type distribution $\ddot{p} \in \hat{\Delta}$ and any deterministic initial condition $\Pi^0=(\alpha^0, \pi^0)$ there exists a {rich} Fubini extension $(I \times \Omega, \cal{I} \boxtimes \mathcal{F}, \lambda \boxtimes P)$ on which a discrete dynamical system $\mathbb{D}=\left(\Pi^{n}\right)_{n=0}^{\infty}$ as in {Definition 3.6 in \cite{biagini_mazzon_oberpriller_meyer_brandis_2022}} can be constructed with {discrete time input processes $(\eta^n,\theta^n,\xi^n, \sigma^n, \varsigma^n)_{n\geq 1}$ coming from $(\eta_{kl},\theta_{kl}, \xi_{kl}, \sigma_{kl}[r,s], \varsigma_{kl}[r])_{k,l,r,s \in S \times S \times S \times S }$ as stated in {Section 2 in \cite{biagini_mazzon_oberpriller_meyer_brandis_2022}}}. In particular, 
$$  \Omega = \tilde{\Omega} \times  \hat{\Omega}, \quad \mathcal{F} = \tilde{\mathcal{F}} \otimes  \hat{\mathcal{F}},\quad  P = \tilde P \ltimes \hat P,$$
where $(\hat \Omega, \hat{\mathcal{F}})$ is a measurable space and $\hat P$ a Markov kernel from $\tilde \Omega$ to $\hat \Omega$. The dynamical system $\mathbb{D}$ is also MCI according to {Definition 3.8 in \cite{biagini_mazzon_oberpriller_meyer_brandis_2022}} and with initial cross-sectional extended type distribution {$\hat{p}^0$ equal to $\ddot{p}^0$ with probability one.} 
\end{theorem}

\begin{proof}
	At each time period we construct three internal measurable spaces with internal transition probabilities taking into account the following steps:
	\begin{enumerate}
		\item random mutation
		\item random matching
		\item random type changing with break-up.
	\end{enumerate}  
	Let $M$ be a limited hyperfinite number in $\leftidx{^*}{\mathbb{N}}_{\infty}$. Let $\lbrace n \rbrace_{n=0}^M$ be the hyperfinite discrete time line and $(I, \cal{I}_0, \lambda_0)$ the agent space, where $I=\lbrace 1,..., \hat{M} \rbrace$, $\cal{I}_0$ is the internal power set on $I$, $\lambda_0$ is the internal counting probability measure on $\cal{I}_0$, and $\hat{M}$ is an unlimited hyperfinite number in $\leftidx{^*}{\mathbb{N}}_{\infty}$. 

We start  by transferring the deterministic functions{\footnote{{Note that at initial time, the functions are supposed to be deterministic and in particular independent of $\tilde{\Omega}$.}} $\eta(0,\cdot),\theta(0,\cdot), \xi(0,\cdot), \sigma(0,\cdot), \varsigma(0,\cdot): \hat{\Delta} \to [0,1]$} to the nonstandard universe. In particular, we denote by $\leftidx{^*}{\theta}^0_{kl}$ for any $k,l \in S$ and by $\leftidx{^*}{f}^0$ for $f=\eta,\xi, \sigma, \varsigma$ the internal functions from $\leftidx{^*}{\hat{\Delta}}$ to $[0,1]$.  We also let $\hat{\theta}^{0}_{kl}(\hat{\rho})=\leftidx{^*}{\hat{\theta}}^{0}_{kl}(\hat{\rho})$ and $\hat{b}_k^0=1-\sum_{l \in S} \hat{\theta}_{kl}^0(\hat{\rho})$ for any $k,l \in S$ and $\hat{\rho} \in \leftidx{^*}{\hat{\Delta}}$, with $1 \in \leftidx{^*}{\mathbb{N}}$.

 We start at $n=0$. To do so, we introduce the trivial probability space over the single set $\lbrace 0 \rbrace$ denoted by $(\bar{\Omega}_0,\bar{\mathcal{F}}_0, \bar{Q}_0)$. Let $\lbrace A_{kl} \rbrace_{(k,l)\in \hat{S}}$ be an internal partition of $I$ such that $\frac{\vert A_{kl} \vert}{\hat{M}} \simeq \ddot{p}_{kl}$ for any $k \in S$ and $l \in S \cup \lbrace J \rbrace$, such that $ \vert A_{kk} \vert$ is even for any $k,l \in S$ and $\vert A_{kl} \vert = \vert A_{lk} \vert$ for any $k,l \in S$. Let $\alpha^0$ be an internal function from $(I, \cal{I}_0, \lambda_0)$ to $S$ such that $\alpha^{0}(i)=k$ if $i \in \bigcup_{l \in S \cup \lbrace J \rbrace} A_{kl}$. Let $\pi^0$ be an internal partial matching from $I$ to $I$ such that $\pi^0(i)=i$ on $\bigcup_{k \in S} A_{kJ}$, and the restriction $\pi^0 \vert_{A_{kl}}$ is an internal bijection from $A_{kl}$ to $A_{lk}$ for any $k,l \in S$. Let 
\[g^0(i)=\begin{cases}
\alpha^0(\pi^0(i)) & \text{ if } \pi^0(i) \neq i\\
J& \text{ if } \pi^0(i) = i.
\end{cases}\]
It is clear that $\lambda_0(\lbrace i: \alpha^0(i)=k, g^0(i)=l \rbrace) \simeq \ddot{p}^0_{kl}$ for any $k \in S$ and $l \in S \cup \lbrace J \rbrace$. \\
Let $(\tilde{\Omega}, \tilde{\mathcal{F}}, \tilde{P})$ be the hyperfinite internal space. Since the intensities are supposed to be deterministic at initial time, the Markov kernel from $\tilde{\Omega}$ is trivial and we define the initial internal product probability space as
 \begin{equation*}
	(\Omega_0, \mathcal{F}_0, Q_0):=(\tilde{\Omega} \times \bar{\Omega}_0, \tilde{\mathcal{F}} \otimes \bar{\mathcal{F}_0}, \tilde{P}\otimes \bar{Q}_0).
\end{equation*}
Suppose now that the dynamical system $\mathbb{D}$ has been constructed up to time $n-1 \in \leftidx{^*}{\mathbb{N}}$ for $n \geq 1$, i.e., that the sequences $\lbrace (\Omega_m, \mathcal{F}_m, {Q}_m) \rbrace_{m=0}^{3n-3}$ and $\lbrace \alpha^l, \pi^l \rbrace_{l=0}^{n-1}$ have been constructed. In particular, we assume to have introduced the spaces $(\hat{\Omega}_m, \hat{\mathcal{F}}_m)$ and the Markov kernel $\hat P_m$ from $\tilde \Omega$ to $\hat \Omega_m$ for any $m=1, \dots, n-3$, so that we can define $\Omega_m:=\tilde \Omega \times \hat{\Omega}_m$ as a hyperfinite internal set with internal power set $\mathcal{F}_m:=\tilde{\mathcal{F}} \otimes \hat{\mathcal{F}}_m$ and $Q_m:= \tilde P \ltimes \hat P _m$ as an internal transition probability from $\Omega^{m-1}$ to $(\Omega_m, \mathcal{F}_m)$, where
\begin{equation} \label{eq:DefiHatOmega}
	\Omega^m := \tilde{\Omega} \times \hat{\Omega}^m, \quad  \hat{\Omega}^m:=\bar{\Omega}_0 \times \prod_{j=1}^{m}  \hat \Omega_j,  \quad  \hat{\mathcal{F}}^m:=\bar{\mathcal{F}}_0 \otimes \left( \otimes_{j=1}^{m} \hat{\mathcal{F}}_j \right) \quad  \text{ and } \quad \mathcal{F}^m = \tilde{\mathcal{F}} \otimes \hat{\mathcal{F}}^m.
\end{equation}         
  
 In this setting, $\alpha^{l}$ is an internal type function from $I \times {\Omega}^{3l-1}$ to the space $S$, and $\pi^l$ an internal random matching from $I \times {\Omega}^{3l}$ to $I$,
such that
 $$
 \alpha^l(i,(\tilde{\omega}, \hat{\omega}^{3l-1}))= \alpha^l(i,\hat{\omega}^{3l-1}), \quad \text{for any } (\tilde{\omega}, \hat{\omega}^{3l-1}) \in \Omega^{3l-1}
 $$ 
 and
 $$
 \pi^l(i,(\tilde{\omega}, \hat{\omega}^{3l}))=\pi^l(i, \hat{\omega}^{3l}), \quad \text{for any } (\tilde{\omega}, \hat{\omega}^{3l}) \in \Omega^{3l}. 
 $$
 Given $\omega^{3l} \in \Omega^{3l}$ we denote by $\pi^l_{\hat{\omega}^{3l}}:I \to I$ the function given by
 $$
 \pi^l_{\hat{\omega}^{3l}}(i):=\pi^l(i,(\tilde{\omega}, \hat{\omega}^{3l}))=\pi^l(i, \hat{\omega}^{3l}).
 $$ A similar notation will be used for $\alpha^l_{\hat{\omega}^{3l}}:I \to S$. We now have the following.  
 
\textbf{(i) Random mutation step:}

We let $\hat \Omega_{3n-2}:=S^I$, which is the space of all internal functions from $I$ to $S$, and denote its internal power set by $\hat{\mathcal{F}}_{3n-2}$. For each $i \in I$ and $\omega^{3n-3}=(\tilde \omega, \hat \omega^{3n-3}) \in \Omega^{3n-3}$, if $\alpha^{n-1}(i, {\omega}^{3n-3})=\alpha^{n-1}\left(i, \hat{\omega}^{3n-3}\right)=k$, define a probability measure $\gamma_{i}^{\tilde \omega, \hat \omega^{3n-3}}$ on $S$ by letting $\gamma_{i}^{\tilde \omega, \hat \omega^{3n-3}}(l):=\theta_{kl}({\tilde{\omega}},n,\hat{\rho}^{n-1}_{\hat{\omega}^{3n-3}})$ for each $l \in S$ with
{
$$
\hat{\rho}^{n-1}_{\hat{\omega}^{3n-3}}[k,r]:=\lambda_0(\lbrace i \in I:\alpha^n_{\hat{\omega}^{3n-3}}(i)=k,  \alpha^n_{\hat{\omega}^{3n-3}}\left(\pi^n_{\hat{\omega}^{3n-3}}(i)\right)=r\rbrace), \quad k,r \in S
$$
and
$$
\hat{\rho}^{n-1}_{\hat{\omega}^{3n-3}}[k,J]:=\lambda_0(\lbrace i \in I:\alpha^n_{\hat{\omega}^{3n-3}}(i)=k, \pi^n_{\hat{\omega}^{3n-3}}(i)=i\rbrace), \quad k \in S.
$$
}

Define a Markov kernel $\hat P^{\hat \omega^{3n-3}}_{3n-2}$ from ${\tilde{\Omega}}$ to $\hat \Omega_{3n-2}$ by letting $\hat P^{\hat\omega^{3n-3}}_{3n-2}(\tilde \omega)$  be the internal product measure $\prod_{i \in I} \gamma_{i}^{\tilde \omega, \hat \omega^{3n-3}}$. Define $\bar{\alpha}^n:\left( I \times {\Omega}^{3n-2} \right) \to S$ by 
$$
\bar{\alpha}^n(i, (\tilde{\omega},\hat{\omega}^{3n-2})):=\bar{\alpha}^n(i, \hat{\omega}^{3n-2})=:\hat{\omega}_{3n-2}(i)
$$ 
{and $\bar{g}^n: \left( I \times  \Omega^{3n-2}\right) \to S \cup \lbrace J \rbrace$ by
\small{
\[\bar{g}^n(i, (\tilde{\omega},\hat{\omega}^{3n-2})):=\bar{g}^n(i, \hat{\omega}^{3n-2}):=\begin{cases}
\bar{\alpha}^n(\pi^{n-1}(i,\hat{\omega}^{3n-3}),\hat{\omega}^{3n-2}) & \text{ if } \pi^{n-1}(i,{ \hat{\omega}^{3n-3}}) \neq i\\
J& \text{ if }  \pi^{n-1}(i, \hat{\omega}^{3n-3})) = i.
\end{cases}\]}
}
Moreover, we introduce the notation
$$\bar{\alpha}_{\hat{\omega}^{3n-2}}^n(\cdot):I \to S, \quad  \bar{\alpha}_{\hat{\omega}^{3n-2}}^n(i):=\bar{\alpha}^n(i, (\tilde{\omega},\hat{\omega}^{3n-2})){:=}\bar{\alpha}^n(i, \hat{\omega}^{3n-2})$$
 for the type function. We then define $\pi_{\hat{\omega}^{3n-3}}^{n-1}(\cdot): I \to I$ and $g^n_{\hat{\omega}^{3n-2}}:I \to S \cup \lbrace J \rbrace$ analogously. 
Finally, we define the cross-internal extended type distribution after random mutation $\check{\rho}^n_{\hat{\omega}^{3n-2}}$ by
$$
\check{\rho}^n_{\hat{\omega}^{3n-2}}[k, l] := \lambda_0(\lbrace \in I: \bar{\alpha}^n_{\hat{\omega}^{3n-2}}(i) = k, \bar{g}^n_{\hat{\omega}^{3n-2}}(i)=l \rbrace), \quad k, l \in S.
$$  
\textbf{(ii) Directed random matching:}\\
Let $(\hat \Omega_{3n-1}, \hat{\mathcal{F}}_{3n-1})$ and $\hat P^{\hat\omega^{3n-2}}_{3n-1}$ be the measurable space and the Markov kernel, respectively, provided by {Proposition 3.12 in \cite{biagini_mazzon_oberpriller_meyer_brandis_2022}}, with type function $\bar{\alpha}_{\hat{\omega}^{3n-2}}^n(\cdot)$ and partial matching function $\pi_{\hat{\omega}^{3n-3}}^{n-1}(\cdot)$, {for fixed matching probability function $\theta\left(\cdot, n, \check{\rho}^{n}_{\hat{\omega}^{3n-2}}\right)$}. 
{Proposition 3.12 in \cite{biagini_mazzon_oberpriller_meyer_brandis_2022}} also provides the directed random matching
\begin{equation*}
	\pi_{\theta^n\left(\cdot, \check{\rho}^{n}_{\hat{\omega}^{3n-2}}\right),\bar{\alpha}^n_{\hat{\omega}^{3n-2}}, \pi^{n-1}_{\hat{\omega}^{3n-3}}},
\end{equation*}
which is a function defined on $(\Omega_{3n-1},\mathcal{F}_{3n-1})$ {by}
$$
\pi_{\theta^n\left(\cdot, \check{\rho}^{n}_{\hat{\omega}^{3n-2}}\right),\bar{\alpha}^n_{\hat{\omega}^{3n-2}}, \pi^{n-1}_{\hat{\omega}^{3n-3}}}(i,(\tilde{\omega},\hat{\omega}_{3n-1})){:=}\pi_{\theta^n\left(\cdot, \check{\rho}^{n}_{\hat{\omega}^{3n-2}}\right),\bar{\alpha}^n_{\hat{\omega}^{3n-2}}, \pi^{n-1}_{\hat{\omega}^{3n-3}}}(i,\hat{\omega}_{3n-1}).
$$
We then define $\bar{\pi}^n: \left( I \times \Omega^{3n-1}\right) \to I$ by 
\begin{equation*}
\bar{\pi}^n(i,(\tilde{\omega},\hat{\omega}_{3n-1})):=\bar{\pi}^n(i,\hat{\omega}^{3n-1}):=\pi_{\theta^n\left(\cdot,\check{\rho}^{n}_{\hat{\omega}^{3n-2}}\right),\bar{\alpha}^n_{\hat{\omega}^{3n-2}}, \pi^{n-1}_{\hat{\omega}^{3n-3}}}(i, \hat \omega_{3n-1})	
\end{equation*}
and
\[\bar{\bar{g}}^n(i, (\tilde{\omega},\hat{\omega}^{3n-1}))=\bar{\bar{g}}^n(i, \hat{\omega}^{3n-1}):=\begin{cases}
\bar{\alpha}^n(\bar{\pi}^n(i, \hat{\omega}^{3n-1}), \hat{\omega}^{3n-2}) & \text{ if } \bar{\pi}^{n}(i, \hat{\omega}^{3n-1})\neq i  \\
J& \text{ if } \bar{\pi}^{n}(i, \hat{\omega}^{3n-1})= i.\end{cases}\]
{Define now  the cross-internal extended type distribution after the random matching $\check{\check{\rho}}^n_{\hat{\omega}^{3n-1}}$ by
$$
\check{\check{\rho}}^n_{\hat{\omega}^{3n-1}}[k,l]:=\lambda_0(\lbrace \in I: \bar{\alpha}^{n}_{\hat{\omega}^{3n-1}}(i)=k, \bar{\bar{g}}^n_{\hat{\omega}^{3n-1}}(i)=l \rbrace).
$$} 

\textbf{(iii) Random type changing with break-up for matched agents:}

Introduce $\hat\Omega_{3n}:=(S \times \lbrace0,1 \rbrace)^{I}$ with internal power set $\hat{\mathcal{F}}_{3n}$, where $0$ represents ``unmatched'' and $1$ represents ``paired''; each point ${\hat\omega_{3n}=(\hat\omega_{3n}^1, \hat\omega_{3n}^2) \in \hat\Omega_{3n}}$ represents an internal function from $I$ to $S \times \lbrace 0,1 \rbrace$. Define a new type function $\alpha^{n}:(I \times \Omega^{3n}) \to S$ by letting $\alpha^n(i,(\tilde{\omega},\hat{\omega}^{3n})):=\alpha^n(i,\hat{\omega}^{3n})=\hat\omega_{3n}^1(i)$. Fix $(\tilde \omega,\hat\omega^{3n-1}) \in \Omega^{3n-1}$. For each $i \in I$, we proceed in the following way.
\begin{enumerate}
	\item If $\bar{\pi}^n(i, \hat{\omega}^{3n-1})=i$ ($i$ is not paired after the matching step at time $n$), let $\tau_i^{\tilde\omega, \hat{\omega}^{3n-1}}$ be the 
probability measure on the type space $S \times \lbrace 0,1 \rbrace$ that gives probability one to the type  $(\bar{\alpha}^n(i,(\tilde{\omega},\hat{\omega}^{3n-2})),0)=(\bar{\alpha}^n(i,\hat{\omega}^{3n-2}),0)$ and zero to the rest
\item If $\bar{\pi}^n(i, (\tilde{\omega},\hat{\omega}^{3n-1}))=\bar{\pi}^n(i, \hat{\omega}^{3n-1})=j\neq i$ ($i$ is paired after the matching step at time $n$), $\bar{\alpha}^n(i,(\tilde{\omega},\hat{\omega}^{3n-2}))=\bar{\alpha}^n(i,\hat{\omega}^{3n-2})=k,  \bar{\pi}^n(i, (\tilde{\omega},\hat{\omega}^{3n-1}))=\bar{\pi}^n(i, \hat{\omega}^{3n-1})=j$ and $\bar{\alpha}^n(j,(\tilde{\omega},\hat{\omega}^{3n-1}))=\bar{\alpha}^n(j,\hat{\omega}^{3n-1})=l$, define a probability measure $\tau_{ij}^{\tilde\omega, \hat{\omega}^{3n-1}}$ on $(S \times \lbrace 0,1 \rbrace) \times(S \times \lbrace 0,1 \rbrace)$ as
\small{
\begin{align*}
	&\tau_{ij}^{\tilde\omega, \hat{\omega}^{3n-1}}((k',0), (l',0))	:=\left(1-\xi_{kl}({\tilde{\omega}},n,\check{\check{\rho}}^n_{\hat{\omega}^{3n-1}})\right)\varsigma_{kl}[k']\left({\tilde{\omega}},n,\check{\check{\rho}}^n_{\hat{\omega}^{3n-1}}\right)\varsigma_{lk}[l']\left({\tilde{\omega}},n,\check{\check{\rho}}^n_{\hat{\omega}^{3n-1}}\right)
\end{align*}}
and \small{
\begin{align*}
	\tau_{ij}^{\tilde\omega, \hat{\omega}^{3n-1}}\left((k',1), (l',1)\right)
	&:=\xi_{kl}\left({\tilde{\omega}},n,\check{\check{\rho}}^n_{\hat{\omega}^{3n-1}}\right)\sigma_{kl}[k',l']\left({\tilde{\omega}},n,\check{\check{\rho}}^n_{\hat{\omega}^{3n-1}}\right)
\end{align*}}
for $k',l' \in S$, and zero for the rest.
\end{enumerate}

Let $A^{n}_{\hat\omega^{3n-1}}=\lbrace (i,j) \in I \times I: i < j,  \bar{\pi}^n(i,(\tilde{\omega},\hat{\omega}^{3n-1}))=\bar{\pi}^n(i,\hat{\omega}^{3n-1})=j \rbrace$ and $B^{n}_{\hat\omega^{3n-1}}=\lbrace i \in I:\bar{\pi}^n(i,(\tilde{\omega},\hat{\omega}^{3n-1}))= \bar{\pi}^n(i,\hat{\omega}^{3n-1})=i \rbrace.$ {Define a Markov kernel $\hat P_{3n}^{\hat \omega^{3n-1}}$ from $\tilde \Omega$ to $\hat \Omega ^{3n}$ by 
\begin{equation*}
	\hat P_{3n}^{\hat \omega^{3n-1}}(\tilde \omega):=\prod_{i \in B^{n}_{\tilde\omega, \hat{\omega}^{3n-1}}} \tau_i^{\tilde\omega, \hat{\omega}^{3n-1}} \otimes \prod_{(i,j) \in A^n_{\hat{\omega}^{3n-1}}} \tau_{ij}^{\tilde \omega, \hat{\omega}^{3n-1}}.
\end{equation*}
Let 
\begin{align*}
\pi^n(i,(\tilde{\omega},\hat{\omega}^{3n}))&={\pi}^n(i, \hat{\omega}^{3n})
\\&:=\begin{cases}
J & \text{ if } \bar{\pi}^{n}(i, \hat{\omega}^{3n-1})=J \text{ or } \blue{\hat{\omega}}_{{3n}}^2(i)=0 \text{ or } \blue{\hat{\omega}}_{{3n}}^2(\bar{\pi}^{n}(i, \hat{\omega}^{3n-1}))=0 \\
\bar{\pi}^{n}(i, \hat{\omega}^{3n-1})& \text{ otherwise,}\end{cases}
\end{align*}
and
\[g^n(i,(\tilde{\omega},\hat{\omega}^{3n}))={g}^n(i, \hat{\omega}^{3n}):=\begin{cases}
\alpha^n(\pi^n(i, \hat{\omega}^{3n}), \hat{\omega}^{3n}) & \text{ if } {\pi}^{n}(i, \hat{\omega}^{3n})\neq i \\
J& \text{ if } {\pi}^{n}(i, \hat{\omega}^{3n})= i. \end{cases}\]
Define $\hat{\rho}_{\hat{\omega}^{3n}}^n=\lambda_0(\alpha^n_{\hat{\omega}^{3n}}, \pi^n_{\hat{\omega}^{3n}} )^{-1}$. \\

By repeating this procedure, we construct a hyperfinite sequence of internal transition probability spaces $\lbrace(\Omega_m, \mathcal{F}_m,Q_m)\rbrace_{m=0}^{3M}$ and a hyperfinite sequence of internal type functions and internal random matchings $\lbrace (\alpha^n, \pi^n) \rbrace_{n=0}^M$. Moreover, define $(\Omega^m,\mathcal{F}^m)$ as in \eqref{eq:DefiHatOmega}, and
$$
\hat P^m := \prod_{i=1}^m \hat P_i , \qquad Q^m:=\tilde P \ltimes \hat P^m,
$$
where the product of the Markov kernels is $\tilde \omega$-wise.
}

Let $(I \times \Omega^{3M}, \cal{I}_0 \otimes \mathcal{F}^{3M}, \lambda_0 \otimes Q^{3M})$ be the internal product probability space of $(I, \cal{I}_0, \lambda_0)$ and $(\Omega^{3M}, \mathcal{F}^{3M}, Q^{3M}).$ Denote the Loeb spaces of $(\Omega^{3M}, \mathcal{F}^{3M}, Q^{3M})$ and the internal product $(I \times \Omega^{3M}, \cal{I}_0 \otimes \mathcal{F}^{3M}, \lambda_0 \otimes Q^{3M})$ by $(\Omega^{3M}, \mathcal{F}, P)$ and $(I \times \Omega^{3M}, \cal{I} \boxtimes \mathcal{F}, \lambda \boxtimes P),$ respectively. For simplicity, let $\Omega^{3M}$ be denoted by $\Omega$ and $\hat \Omega^{3M}$ by $\hat \Omega$. Denote now $Q^{3M}$ by $P$ and the Markov kernel $\hat P^{3M}$ by $\hat P$.

The properties of a dynamical system as well as the independence conditions follow now by applying similar arguments as in the proof of Theorem 5 in \cite{RandomMatchingDiscrete} for any fixed $\tilde \omega \in \tilde \Omega$. The only difference is that in our setting the {input processes} for the random mutation step and the break-up step also depend on the extended type distribution. Furthermore, these arguments are similar to the ones in the proof of Lemma \ref{lemma:Auxiliary1} and can be found there with all details. 
\end{proof}

\section{Proof of {Theorem 3.14 in \cite{biagini_mazzon_oberpriller_meyer_brandis_2022}}}

We now prove {Theorem 3.14 in \cite{biagini_mazzon_oberpriller_meyer_brandis_2022}} which is a generalization of the results in Appendix C in \cite{RandomMatchingDiscrete}. For $n \geq 1$ we define the mapping $\Gamma^n$ from ${\tilde{\Omega}} \times \hat{\Delta} $ to $ \hat{\Delta} $ by
\begin{align}
	\Gamma^n_{kl}({\tilde{\omega}},\hat{p}) &=\sum_{k_1, l_1 \in S} (1-\xi_{k_1 l_1}({\tilde{\omega}},n,\tilde{\tilde{p}}^n)) \sigma_{k_1 l_1}[k,l]\left({\tilde{\omega}},n,\tilde{\tilde{p}}^n\right) \tilde{p}^n_{k_1l_1} \nonumber  \\
	&\quad+\sum_{k_1, l_1 \in S} (1-\xi_{k_1 l_1}({\tilde{\omega}},n,\tilde{\tilde{p}}^n)) \sigma_{k_1 l_1}[k,l]\left({\tilde{\omega}},n,\tilde{\tilde{p}}^n\right) {\theta}_{k_1l_1}({\tilde{\omega}},n,\tilde{p}^n)\tilde{p}^n_{k_1J}, \label{eq:DefinitionGamma}
\end{align}
and 
{
\begin{align}
\Gamma^n_{kJ}({\tilde{\omega}},\hat{p})&=  b_{k}(\tilde{\omega},n,\tilde{p}^n)\tilde{p}_{kJ}^n + \sum_{k_1,l_1 \in S} \xi_{k_1l_1}(\tilde{\omega},n,\tilde{\tilde{p}}^n) \varsigma_{k_1l_1}[k](\tilde{\omega}, \tilde{\tilde{p}}^n)\tilde{p}^n_{k_1l_1} \nonumber \\
&+\sum_{k_1, l_1 \in S} \xi_{k_1 l_1  }(\tilde{\omega},n, \tilde{\tilde{p}}^n) \varsigma_{k_1l_1}[k](\tilde{\omega},n,\tilde{\tilde{p}}^n) \theta_{k_1 l_1}(\tilde{\omega},n,\tilde{p}^n)\tilde{p}_{k_1J}^n \label{eq:x}
\end{align}}
with
\begin{align*}
	\tilde{p}_{kl}^n&= \sum_{k_1,l_1 \in S } \eta_{k_1 k}({\tilde{\omega}},n,\hat{p}) \eta_{l_1 l}({\tilde{\omega}},n,\hat{p})\hat{p}_{k_1l_1} \\
	\tilde{p}_{kJ}^n&=\sum_{l \in S} \hat{p}_{lJ} \eta_{lk}({\tilde{\omega}},n,\hat{p})
\end{align*}
and 
\begin{align*}
	\tilde{\tilde{p}}_{kl}^n&=\tilde{p}^n_{kl}+\theta_{kl}({\tilde{\omega}},n,\tilde{p}^n)\tilde{p}^n_{kJ} \\
	\tilde{\tilde{p}}_{kJ}^n&=b_k({\tilde{\omega}},n,\tilde{p}^n) \tilde{p}^n_{kJ}.
\end{align*}

{Theorem 3.14 in \cite{biagini_mazzon_oberpriller_meyer_brandis_2022}} is proven with the help of the following lemmas. 
\begin{lemma} \label{lemma:IndependenceRandomProcesses}
	Assume that the {discrete} dynamical system {$\mathbb{D}$} defined in {Definition 3.6 in \cite{biagini_mazzon_oberpriller_meyer_brandis_2022}} is Markov conditionally independent given {$\tilde{\omega}$} as defined in {Definition 3.8 in \cite{biagini_mazzon_oberpriller_meyer_brandis_2022}}. Then given $\tilde{\omega} \in \tilde{\Omega}$, the discrete time processes $\lbrace \beta_i^n \rbrace_{n=0}^{\infty}, i \in I,$ are essentially pairwise independent on $(I \times \hat{\Omega}, \cal{I} \boxtimes \hat{\mathcal{F}}, \lambda \boxtimes \hat{P}^{\tilde{\omega}})$. Moreover, for fixed $n=1,...,M$ also $(\bar{\beta}_i^n)_{n=0}^{\infty}$ and $(\bar{\bar{\beta}}_i^n)_{n=0}^{\infty}, i \in I$, are essentially pairwise independent on $(I \times \hat{\Omega}, \cal{I} \boxtimes \hat{\mathcal{F}}, \lambda \boxtimes \hat{P}^{\tilde{\omega}})$. 
	\end{lemma}
\begin{proof}
	This can be proven by the same arguments used in the proof of Lemma 3 in \cite{RandomMatchingDiscrete}.
\end{proof}
We now derive a result which shows how to compute for a fixed $\tilde{\omega} \in \tilde{\Omega}$ the expected cross-sectional distributions $\mathbb{E}^{\hat{P}^{\tilde{\omega}}}[\check{p}^n ]$, $\mathbb{E}^{\hat{P}^{\tilde{\omega}}}[\check{\check{p}}^n ]$ and $\mathbb{E}^{\hat{P}^{\tilde{\omega}}}[\hat{p}^n]$.
\begin{lemma} \label{lemma:Auxiliary1}
	The following holds for any fixed $\tilde{\omega} \in \tilde{\Omega}$.
	\begin{enumerate}
		\item For each $n \geq 1$, $\mathbb{E}^{\hat{P}^{\tilde{\omega}}}[\hat{p}^n]= \Gamma^n ({\tilde{\omega}},\mathbb{E}^{\hat{P}^{\tilde{\omega}}}[\hat{p}^{n-1}])$, with $\Gamma$ defined in \eqref{eq:DefinitionGamma}.
		\item For each $n\geq 1$, we have
		\begin{align*}
		\mathbb{E}^{\hat{P}^{\tilde{\omega}}}[\check{p}^n_{kl} ]=\sum_{k_1, l_1 \in S}  {\eta}_{k_1,k}({\tilde{\omega}},n,\mathbb{E}^{\hat{P}^{\tilde{\omega}}}[\hat{p}^{n-1}]){\eta}_{l_1,l}({\tilde{\omega}},n,\mathbb{E}^{\hat{P}^{\tilde{\omega}}}[\hat{p}^{n-1}]) \mathbb{E}^{\hat{P}^{\tilde{\omega}}}[\hat{p}^{n-1}_{k_1,l_1}]
	\end{align*}
	and
	\begin{align} \nonumber
		\mathbb{E}^{\hat{P}^{\tilde{\omega}}}[\check{p}_{kJ}^n]=\sum_{k_1 \in S} {\eta}_{k_1,k}({\tilde{\omega}},n,\mathbb{E}^{\hat{P}^{\tilde{\omega}}}[\hat{p}^{n-1}])  \mathbb{E}^{\hat{P}^{\tilde{\omega}}}[\hat{p}^{n-1}_{k_1,J}].
	\end{align}
	\item For each $n \geq 1$, we have
	\begin{align} \nonumber
		\mathbb{E}^{\hat{P}^{\tilde{\omega}}}[\check{\check{p}}^n_{kl}] =\mathbb{E}^{\hat{P}^{\tilde{\omega}}}[\check{p}^n_{kl}]+ {\theta}_{kl}({\tilde{\omega}},n,\mathbb{E}^{\hat{P}^{\tilde{\omega}}}[\check{p}^n ])\mathbb{E}^{\hat{P}^{\tilde{\omega}}}[\check{p}^n_{kJ}]
	\end{align}
	and
	\begin{align} \nonumber
	\mathbb{E}^{\hat{P}^{\tilde{\omega}}}[\check{\check{p}}^n_{kJ} ] 
={b}_{k}({\tilde{\omega}},n,\mathbb{E}^{\hat{P}^{\tilde{\omega}}}[\check{p}^n ]) \mathbb{E}^{\hat{P}^{\tilde{\omega}}}[\check{p}^n_{kJ}].
\end{align}
	\end{enumerate}	
\end{lemma}
\begin{proof}
Fix $\tilde{\omega} \in \tilde{\Omega}$ and $k,l \in S$. By Lemma \ref{lemma:IndependenceRandomProcesses} we know that the processes $(\beta_i^n)_{n=0}^{\infty}, i \in I,$ are essentially pairwise independent. Then the exact law of large numbers in Lemma 1 in \cite{RandomMatchingDiscrete} implies that ${\hat{p}^{n-1}(\hat{\omega})}=\mathbb{E}^{\hat{P}^{\tilde{\omega}}}[\lambda(\beta^{n-1})^{-1}]$ for $\hat{P}$-almost all $\hat{\omega} \in \hat{\Omega}$. Thus equations \eqref{eq:IndiMutation1} and \eqref{eq:IndiMutation2} are equivalent to 
\begin{align}
&\hat{P}^{\tilde{\omega}}\left(\bar{\alpha}_i^n=k_2, \bar{g}_i^n=l_2 \vert \alpha_{i}^{n-1}=k_1, g_i^{n-1}=l_1  \right)={\eta}_{k_1,k_2}\left({\tilde{\omega}},n,\mathbb{E}^{P^{\tilde{\omega}}}[\hat{p}^{n-1}]\right){\eta}_{l_1,l_2}\left({\tilde{\omega}},n,\mathbb{E}^{P^{\tilde{\omega}}}[\hat{p}^{n-1}]\right)  \label{eq:IndiMutatationRewritten1.0}\\
&{\hat P^{\tilde{\omega}}}\left(\bar{\alpha}_i^n=k_2, \bar{g}_i^n=r \vert \alpha_{i}^{n-1}=k_1, g_i^{n-1}=J\right)={\eta}_{k_1,k_2}\left({\tilde{\omega}},n,\mathbb{E}^{P^{\tilde{\omega}}}[\hat{p}^{n-1}]\right)\delta_J(r).  \label{eq:IndiMutatationRewritten2.0}
\end{align} Therefore, for any $k_1,l_1 \in S$ we have
\begin{align}
	&\hat{P}^{\tilde{\omega}}\left(\bar{\beta}_i^n=(k,J) \vert \beta_i^{n-1}=(k_1,l_1)\right)=0 \label{eq:GeneralizationLemma4_1} \\
	&\hat{P}^{\tilde{\omega}}\left(\bar{\beta}_i^n=(k,l) \vert \beta_i^{n-1}=(k_1,J) \right)=0. \label{eq:GeneralizationLemma4_2}
\end{align}
Then with the same calculations as in the proof of Lemma 4 in \cite{RandomMatchingDiscrete} we get  that 
\begin{align}
		\mathbb{E}^{\hat{P}^{\tilde{\omega}}}[\check{p}^n_{kl}] &= \mathbb{E}^{\hat{P}^{\tilde{\omega}}}[\lambda(i \in I: \bar{\beta}_{\omega}^n(i)=(k,l)) ]\nonumber \\
		&=\int_I \hat{P}^{\tilde{\omega}}( \bar{\beta}^n_i=(k,l)) d\lambda(i) \nonumber \\
		&=\sum_{k_1,l_1 \in S} {\eta}_{k_1,k}({\tilde{\omega}},n,\mathbb{E}^{\hat{P}^{\tilde{\omega}}}[\hat{p}^{n-1}])  {\eta}_{l_1,l}({\tilde{\omega}},n,\mathbb{E}^{\hat{P}^{\tilde{\omega}}}[\hat{p}^{n-1}]) \mathbb{E}^{\hat{P}^{\tilde{\omega}}}[\hat{p}^{n-1}_{k_1 l_1}] \label{eq:GeneralizationLemma4_3} 
	\end{align}
and
	\begin{align} \label{eq:GeneralizationLemma4_8}
		\mathbb{E}^{\hat{P}^{\tilde{\omega}}}[\check{p}_{kJ}^n]=\sum_{k_1 \in S} {\eta}_{k_1,k}({\tilde{\omega}},n,\mathbb{E}^{\hat{P}^{\tilde{\omega}}}[\hat{p}^{n-1}])  \mathbb{E}^{\hat{P}^{\tilde{\omega}}}[\hat{p}^{n-1}_{kJ}].
	\end{align}
By Lemma \ref{lemma:IndependenceRandomProcesses} we know that $\bar{\beta}^n$ is  essentially pairwise independent. Again it follows by the exact law of large numbers that ${\check{p}^{n}}(\hat{\omega})=\mathbb{E}^{\hat{P}^{\tilde{\omega}}}[\check{p}^n]$ for $\hat{P}^{\tilde{\omega}}$-almost all $\hat{\omega} \in \hat{\Omega}$.
Then \eqref{eq:MatchingCondProb1} and \eqref{eq:MatchingCondProb2} are equivalent to
\begin{equation} \label{eq:MatchingCondRewrittenProb1}
	\hat{P}^{\tilde{\omega}}(\bar{\bar{g}}^n=l \vert \bar{\alpha}_i^n=k, \bar{g}_i^n=J)={\theta}_{kl}({\tilde{\omega}},n,\mathbb{E}^{\hat{P}^{\tilde{\omega}}}[\check{p}^n])
\end{equation}
\begin{equation} \label{eq:MatchingCondRewrittenProb2}
	\hat{P}^{\tilde{\omega}}(\bar{\bar{g}}^n=J \vert \bar{\alpha}_i^n=k, \bar{g}_i^n=J)={b}_{k}({\tilde{\omega}},n,\mathbb{E}^{\hat{P}^{\tilde{\omega}}}[\check{p}^n ]).
\end{equation}
	
By the same calculations as in the proof of Lemma 4 in \cite{RandomMatchingDiscrete} we have 
	\begin{align}
		&\mathbb{E}^{\hat{P}^{\tilde{\omega}}}[\check{\check{p}}^n_{kl} ]=\mathbb{E}^{\hat{P}^{\tilde{\omega}}}[\check{p}^n_{kl}]+ {\theta}_{kl}({\tilde{\omega}},n,\mathbb{E}^{\hat{P}^{\tilde{\omega}}}[\check{p}^n])\mathbb{E}^{\hat{P}^{\tilde{\omega}}}[\check{p}^n_{kJ}]\label{eq:GeneralizationLemma4_8_a_0}
	\end{align}
and
\begin{align}
	&\mathbb{E}^{\hat{P}^{\tilde{\omega}}}[\check{\check{p}}^n_{kJ}]={b}_{k}({\tilde{\omega}},n,\mathbb{E}^{\hat{P}^{\tilde{\omega}}}[\check{p}^n ]) \mathbb{E}^{\hat{P}^{\tilde{\omega}}}[\check{p}^n_{kJ}]\label{eq:GeneralizationLemma4_8_a}.
\end{align}
By Lemma \ref{lemma:IndependenceRandomProcesses}, $\bar{\bar{\beta}}^n$ is essentially pairwise independent and thus ${\check{\check{p}}^{n}}(\hat{\omega})=\mathbb{E}^{\hat{P}^{\tilde{\omega}}}[{\check{\check{p}}^{n}}]$ for $\hat{P}$-almost all $\hat{\omega} \in \hat{\Omega}$. Then \eqref{eq:BreakUpCondProb1} and \eqref{eq:BreakUpCondProb2} are equivalent to
\begin{align*}
	&\hat{P}^{\tilde{\omega}}(\alpha_i^n=l_1,g_i^n=l_2 \vert \alpha_i^n=k_1, \bar{\bar{g}}_i^n=k_2)=(1-\xi_{k_1k_2}({\tilde{\omega}},n,\mathbb{E}^{\hat{P}^{\tilde{\omega}}}[\check{\check{p}}^n ])) \sigma_{k_1k_2}[l_1,l_2]\left({\tilde{\omega}},n,\mathbb{E}^{\hat{P}^{\tilde{\omega}}}[\check{\check{p}}^n]\right)
\end{align*} 
and
\begin{align*}
	&\hat{P}^{\tilde{\omega}}(\alpha_i^n=l_1,g_i^n=J \vert \alpha_i^n=k_1, \bar{\bar{g}}_i^n=k_2)=\xi_{k_1k_2}({\tilde{\omega}},n,\mathbb{E}^{\hat{P}^{\tilde{\omega}}}[\check{\check{p}}^n ]) \varsigma_{k_1k_2}[l_1,l_2]\left({\tilde{\omega}},n,\mathbb{E}^{\hat{P}^{\tilde{\omega}}}[\check{\check{p}}^n ]\right),
\end{align*}
respectively. Thus 
\begin{align}
	&\mathbb{E}^{\hat{P}^{\tilde{\omega}}}[\hat{p}_{kl}^n]=\sum_{k_1, l_1 \in S} (1-\xi_{k_1 l_1}({\tilde{\omega}},n,\mathbb{E}^{\hat{P}^{\tilde{\omega}}}[\check{\check{p}}^n])) \sigma_{k_1 l_1}[k,l]\left({\tilde{\omega}},n,\mathbb{E}^{\hat{P}^{\tilde{\omega}}}[\check{\check{p}}^n]\right) \mathbb{E}^{\hat{P}^{\tilde{\omega}}}[\check{\check{p}}^n_{k_1 l_1}] \label{eq:GeneralizationLemma4_9}
\end{align}
and
\begin{align}
	\mathbb{E}^{\hat{P}^{\tilde{\omega}}}[\hat{p}_{kJ}^n ] &=\mathbb{E}^{\hat{P}^{\tilde{\omega}}}[\check{\check{p}}^n_{kJ} ] \nonumber \\
	& \quad + \sum_{k_1,l_1 \in S} \xi_{k_1 l_1}({\tilde{\omega}},n,\mathbb{E}^{\hat{P}^{\tilde{\omega}}}[\check{\check{p}}^n]) \varsigma_{k_1 l_1}[k]\left({\tilde{\omega}},n,\mathbb{E}^{\hat{P}^{\tilde{\omega}}}[\check{\check{p}}^n ]\right) \mathbb{E}^{\hat{P}^{\tilde{\omega}}}[\check{\check{p}}^n_{k_1l_1}]. \label{eq:GeneralizationLemma4_10}
\end{align}
By plugging \eqref{eq:GeneralizationLemma4_8} in \eqref{eq:GeneralizationLemma4_9} we get
\begin{align}
	&\mathbb{E}^{\hat{P}^{\tilde{\omega}}}[\hat{p}_{kl}^n] \nonumber \\
	&=\sum_{k_1, l_1 \in S} (1-\xi_{k_1 l_1}({\tilde{\omega}},n,\mathbb{E}^{\hat{P}^{\tilde{\omega}}}[\check{\check{p}}^n])) \sigma_{k_1 l_1}[k,l]\left({\tilde{\omega}},n,\mathbb{E}^{\hat{P}^{\tilde{\omega}}}[\check{\check{p}}^n]\right) \mathbb{E}^{\hat{P}^{\tilde{\omega}}}[{\check{p}}^n_{k_1 l_1}] \nonumber  \\
	&\quad+\sum_{k_1, l_1 \in S} (1-\xi_{k_1 l_1}({\tilde{\omega}},n,\mathbb{E}^{\hat{P}^{\tilde{\omega}}}[\check{\check{p}}^n])) \sigma_{k_1 l_1}[k,l]\left({\tilde{\omega}},n,\mathbb{E}^{\hat{P}^{\tilde{\omega}}}[\check{\check{p}}^n]\right) \eta_{k_1l_1}({\tilde{\omega}},n,\mathbb{E}^{\hat{P}^{\tilde{\omega}}}[{\check{p}}^n])\mathbb{E}^{\hat{P}^{\tilde{\omega}}}[{\check{p}}^n_{k_1J}]. \label{eq:GeneralizationLemma4_10}
\end{align}
By using \eqref{eq:GeneralizationLemma4_8_a} and \eqref{eq:GeneralizationLemma4_9}, it follows that
\begin{align}
	\mathbb{E}^{\hat{P}^{\tilde{\omega}}}[\hat{p}_{kJ}^n ] &={b}_{k}({\tilde{\omega}},n,\mathbb{E}^{\hat{P}^{\tilde{\omega}}}[\check{p}^n ]) \mathbb{E}^{\hat{P}^{\tilde{\omega}}}[\check{p}^n_{kJ}] \nonumber \\
	& \quad + \sum_{k_1,l_1 \in S} \xi_{k_1 l_1}({\tilde{\omega}},n,\mathbb{E}^{\hat{P}^{\tilde{\omega}}}[\check{\check{p}}^n]) \varsigma_{k_1 l_1}[k]\left({\tilde{\omega}},\mathbb{E}^{\hat{P}^{\tilde{\omega}}}[\check{\check{p}}^n ]\right) \mathbb{E}^{\hat{P}^{\tilde{\omega}}}[{\check{p}}^n_{k_1 l_1}] \nonumber \\ 
	& \quad + \sum_{k_1,l_1 \in S} \xi_{k_1 l_1}(\tilde{\omega},n,\mathbb{E}^{\hat{P}^{\tilde{\omega}}}[\check{\check{p}}^n]) \varsigma_{k_1 l_1}[k]\left({\tilde{\omega}},n,\mathbb{E}^{\hat{P}^{\tilde{\omega}}}[\check{\check{p}}^n ]\right){\theta}_{k_1l_1}({\tilde{\omega}},n,\mathbb{E}^{\hat{P}^{\tilde{\omega}}}[{\check{p}}^n])\mathbb{E}^{\hat{P}^{\tilde{\omega}}}[{\check{p}}^n_{k_1J}]. \label{eq:GeneralizationLemma4_11}
	\end{align}
\end{proof}

\begin{lemma} \label{lemma:Auxiliary2}
	Assume that the {discrete} dynamical system {$\mathbb{D}$} defined in {Definition 3.6 in \cite{biagini_mazzon_oberpriller_meyer_brandis_2022}} is Markov conditionally independent given {$\tilde{\omega}\in \tilde{\Omega}$} according to Definition {Definition 3.8 in \cite{biagini_mazzon_oberpriller_meyer_brandis_2022}}. Then for fixed $\tilde{\omega} \in \tilde{\Omega}$ the following holds:
	\begin{enumerate}
		\item For $\lambda$-almost all $i \in I$, the extended type process $\lbrace \beta^n_i \rbrace_{n=0}^{\infty}$ for agent $i$ is a Markov chain on $(I \times \hat{\Omega}, \cal{I} \boxtimes \hat{\mathcal{F}}, \lambda \boxtimes \hat{P}^{\tilde{\omega}})$ with transition matrix $z^n$ after time $n-1$.
		\item $\lbrace \beta^n \rbrace_{n=0}^{\infty}$ is also a Markov chain with transition matrix $z^n$ at time $n-1$.
	\end{enumerate}
\end{lemma}
\begin{proof}
Fix $\tilde{\omega} \in \tilde{\Omega}$.\\
1. The Markov property of $\lbrace \beta^n_i \rbrace_{n=0}^{\infty}$ on $(I \times \hat{\Omega}, \cal{I} \boxtimes \hat{\mathcal{F}}, \lambda \boxtimes \hat{P}^{\tilde{\omega}})$ follows by using the same arguments as in the proof of Lemma 5 in \cite{RandomMatchingDiscrete}, for $\lambda$-almost all $i \in I$. We now derive the transition matrix with similar arguments as in \cite{RandomMatchingDiscrete}.
By putting together \eqref{eq:GeneralizationLemma4_3}, \eqref{eq:GeneralizationLemma4_8} and \eqref{eq:GeneralizationLemma4_10}, we get
\begin{align}
	&\mathbb{E}^{\hat{P}^{\tilde{\omega}}}[\hat{p}_{kl}^n] \nonumber \\
	&=\sum_{k_1, l_1,k',l' \in S} (1-\xi_{k_1 l_1}({\tilde{\omega}},n,\tilde{\tilde{p}}^{\tilde{\omega},n})) \sigma_{k_1 l_1}[k,l]\left({\tilde{\omega}},n,\tilde{\tilde{p}}^{\tilde{\omega},n}\right){\eta}_{k' k_1}({\tilde{\omega}},n,\mathbb{E}^{\hat{P}^{\tilde{\omega}}}[\hat{p}^{n-1}]) \nonumber \\
	&\quad \quad \quad \quad  \cdot {\eta}_{l' l_1}({\tilde{\omega}},n,\mathbb{E}^{\hat{P}^{\tilde{\omega}}}[\hat{p}^{n-1}]) \mathbb{E}^{\hat{P}^{\tilde{\omega}}}[\hat{p}^{n-1}_{k'l'}] \nonumber  \\
	&\quad+\sum_{k_1, l_1,k' \in S} (1-\xi_{k_1 l_1}({\tilde{\omega}},n,\tilde{\tilde{p}}^{\tilde{\omega},n})) \sigma_{k_1 l_1}[k,l]\left({\tilde{\omega}},n,\tilde{\tilde{p}}^{\tilde{\omega},n}\right) {\theta}_{k_1l_1}({\tilde{\omega}},n,\tilde{p}^{\tilde{\omega},n}) \nonumber \\
	&\quad \quad \quad \quad  \cdot {\eta}_{k' k_1}({\tilde{\omega}},n,\mathbb{E}^{\hat{P}^{\tilde{\omega}}}[\hat{p}^{n-1}]) \mathbb{E}^{\hat{P}^{\tilde{\omega}}}[\hat{p}^{n-1}_{k'J}]. \nonumber  
\end{align}	
Thus we have
\begin{align}
	z^n_{(k'J)(kl)}(\tilde \omega)&= \sum_{k_1, l_1 \in S}(1-\xi_{k_1 l_1}({\tilde{\omega}},n,\tilde{\tilde{p}}^{\tilde{\omega},n})) \sigma_{k_1 l_1}[k,l]\left({\tilde{\omega}},n,\tilde{\tilde{p}}^{\tilde{\omega},n}\right) {\theta}_{k_1l_1}({\tilde{\omega}},n,\tilde{p}^{\tilde{\omega},n}) \nonumber \\
	 &\quad \cdot {\eta}_{k' k_1}({\tilde{\omega}},n,\mathbb{E}^{\hat{P}^{\tilde{\omega}}}[\hat{p}^{n-1}]) \label{eq:CalculationTransition1}
\end{align}
and
\begin{align}
	z^n_{(k'l')(kl)}(\tilde \omega)&=\sum_{k_1, l_1 \in S} (1-\xi_{k_1 l_1}({\tilde{\omega}},n,\tilde{\tilde{p}}^{\tilde{\omega},n})) \sigma_{k_1 l_1}[k,l]\left({\tilde{\omega}},n,\tilde{\tilde{p}}^{\tilde{\omega},n}\right){\eta}_{k' k_1}({\tilde{\omega}},n,\mathbb{E}^{\hat{P}^{\tilde{\omega}}}[\hat{p}^{n-1}]) \nonumber \\
	&\quad \quad \quad \quad  \cdot {\eta}_{l' l_1}({\tilde{\omega}},n,\mathbb{E}^{\hat{P}^{\tilde{\omega}}}[\hat{p}^{n-1}]). \label{eq:CalculationTransition2}
\end{align}
Similarly, equations \eqref{eq:GeneralizationLemma4_3}, \eqref{eq:GeneralizationLemma4_8} and \eqref{eq:GeneralizationLemma4_11} yield to
\begin{align}
	\mathbb{E}^{\hat{P}^{\tilde{\omega}}}[\hat{p}_{kJ}^n ] &=\sum_{k' \in S}{b}_{k}({\tilde{\omega}},n,\tilde{p}^{\tilde{\omega},n}) {\eta}_{k'k}({\tilde{\omega}},n,\mathbb{E}^{\hat{P}^{\tilde{\omega}}}[\hat{p}^{n-1}]) \mathbb{E}^{\hat{P}^{\tilde{\omega}}}[\hat{p}^{n-1}_{k'J}]\nonumber \\
	& \quad + \sum_{k_1,l_1,k',l' \in S} \xi_{k_1 l_1}({\tilde{\omega}},n,\tilde{\tilde{p}}^{\tilde{\omega},n}) \varsigma_{k_1 l_1}[k]\left[{\tilde{\omega}},n,\tilde{\tilde{p}}^{\tilde{\omega},n}\right] \nonumber  \\
	& \quad \quad \cdot {\eta}_{k'k_1}({\tilde{\omega}},n,\mathbb{E}^{\hat{P}^{\tilde{\omega}}}[\hat{p}^{n-1}]){\eta}_{l'l_1}({\tilde{\omega}},n,\mathbb{E}^{\hat{P}^{\tilde{\omega}}}[\hat{p}^{n-1}]) \mathbb{E}^{\hat{P}^{\tilde{\omega}}}[\hat{p}^{n-1}_{k'l'}]\nonumber \\ 
	& \quad + \sum_{k_1,l_1,k' \in S} \xi_{k_1 l_1}({\tilde{\omega}},n,\tilde{\tilde{p}}^{\tilde{\omega},n}) \varsigma_{k_1 l_1}[k]\left({\tilde{\omega}},n,\tilde{\tilde{p}}^{\tilde{\omega},n}\right){\theta}_{k_1l_1}({\tilde{\omega}},n,\tilde{p}^{\tilde{\omega},n})\nonumber  \\
	&\quad \quad \cdot  {\eta}_{k'k}({\tilde{\omega}},n,\mathbb{E}^{\hat{P}^{\tilde{\omega}}}[\hat{p}^{n-1}]) \mathbb{E}^{\hat{P}^{\tilde{\omega}}}[\hat{p}^{n-1}_{k'J}]. \nonumber
\end{align}
Therefore, the transition probabilities from time $n-1$ to time $n$ can be written as
\begin{align}
	z^n_{(k'l')(kJ)}(\tilde \omega)&=\sum_{k_1,l_1 \in S} \xi_{k_1 l_1}({\tilde{\omega}},n,\tilde{\tilde{p}}^{\tilde{\omega},n}) \varsigma_{k_1 l_1}[k]\left({\tilde{\omega}},n,\tilde{\tilde{p}}^{\tilde{\omega},n}\right) \nonumber  \\
	& \quad \quad \cdot {\eta}_{k'k_1}({\tilde{\omega}},n,\mathbb{E}^{\hat{P}^{\tilde{\omega}}}[\hat{p}^{n-1}]){\eta}_{l'l_1}({\tilde{\omega}},n,\mathbb{E}^{\hat{P}^{\tilde{\omega}}}[\hat{p}^{n-1}]) \label{eq:CalculationTransition3}
\end{align}
and
\begin{align}
	z_{(k'J)(kJ)}^n(\tilde \omega)&={b}_{k}({\tilde{\omega}},n,\tilde{p}^{\tilde{\omega},n}) {\eta}_{k'k}({\tilde{\omega}},n,\mathbb{E}^{\hat{P}^{\tilde{\omega}}}[\hat{p}^{n-1}]) \nonumber \\
	 &\quad +\sum_{k_1,l_1\in S} \xi_{k_1 l_1}(\tilde{\omega},n,\tilde{\tilde{p}}^{\tilde{\omega},n}) \varsigma_{k_1 l_1}[k]\left({\tilde{\omega}},n,\tilde{\tilde{p}}^{\tilde{\omega},n}\right){\theta}_{k_1l_1}({\tilde{\omega}},n,\tilde{p}^{\tilde{\omega},n}) {\eta}_{k'k}({\tilde{\omega}},n,\mathbb{E}^{\hat{P}^{\tilde{\omega}}}[\hat{p}^{n-1}]). \label{eq:CalculationTransition4}
\end{align}
2. The transition matrix of $\lbrace \beta^n \rbrace_{n=0}^{\infty}$ at time $n-1$ can be derived by using \eqref{eq:CalculationTransition1}-\eqref{eq:CalculationTransition4} and the Fubini property applied to $\lambda \boxtimes \hat{P}^{\tilde{\omega}}$ for every fixed $\tilde{\omega} \in \tilde{\Omega}$ as in the proof of Lemma 6 in \cite{RandomMatchingDiscrete}.
\end{proof}

We are now able to prove {Theorem 3.14 in \cite{biagini_mazzon_oberpriller_meyer_brandis_2022}}, which we present here.

\begin{theorem}
	Assume that the discrete dynamical system $\mathbb{D}$ introduced in {Definition 3.6 in \cite{biagini_mazzon_oberpriller_meyer_brandis_2022}} is Markov conditionally independent given $\tilde{\omega} \in \tilde{\Omega}$ according to {Definition 3.8 in \cite{biagini_mazzon_oberpriller_meyer_brandis_2022}}. Given $\tilde{\omega} \in \tilde{\Omega}$, the following holds:
\begin{enumerate}
	\item For each $n \geq 1$, $\mathbb{E}^{\hat{P}^{\tilde{\omega}}}[\hat{p}^n]= \Gamma^n ({\tilde{\omega}},\mathbb{E}^{\hat{P}^{\tilde{\omega}}}[\hat{p}^{n-1}])$.
	\item For each $n \geq 1 $,  we have
		\begin{align*}
		\mathbb{E}^{\hat{P}^{\tilde{\omega}}}[\check{p}^n_{kl} ]=\sum_{k_1, l_1 \in S}  \eta_{k_1,k}({\tilde{\omega}},n,\mathbb{E}^{\hat{P}^{\tilde{\omega}}}[\hat{p}^{n-1}])\eta_{l_1,l}({\tilde{\omega}},n,\mathbb{E}^{\hat{P}^{\tilde{\omega}}}[\hat{p}^{n-1}]) \mathbb{E}^{\hat{P}^{\tilde{\omega}}}[\hat{p}^{n-1}_{k_1,l_1}]
	\end{align*}
	and
	\begin{align} \nonumber
		\mathbb{E}^{\hat{P}^{\tilde{\omega}}}[\check{p}_{kJ}^n]=\sum_{k_1 \in S} \eta_{k_1,k}({\tilde{\omega}},n,\mathbb{E}^{\hat{P}^{\tilde{\omega}}}[\hat{p}^{n-1}])  \mathbb{E}^{\hat{P}^{\tilde{\omega}}}[\hat{p}^{n-1}_{k_1,J}].
	\end{align}
	\item For each $n \geq 1$,  we have
	\begin{align} \nonumber
		\mathbb{E}^{\hat{P}^{\tilde{\omega}}}[\check{\check{p}}^n_{kl}] =\mathbb{E}^{\hat{P}^{\tilde{\omega}}}[\check{p}^n_{kl}]+ {\theta}_{kl}({\tilde{\omega}},n,\mathbb{E}^{\hat{P}^{\tilde{\omega}}}[\check{p}^n ])\mathbb{E}^{\hat{P}^{\tilde{\omega}}}[\check{p}^n_{kJ}]
	\end{align}
	and
	\begin{align} \nonumber
	\mathbb{E}^{\hat{P}^{\tilde{\omega}}}[\check{\check{p}}^n_{kJ} ] 
		={b}_{k}({\tilde{\omega}},n,\mathbb{E}^{\hat{P}^{\tilde{\omega}}}[\check{p}^n ]) \mathbb{E}^{\hat{P}^{\tilde{\omega}}}[\check{p}^n_{kJ}].
	\end{align}
\item For $\lambda$-almost every agent $i$, the extended-type process $\lbrace \beta_i^n \rbrace_{n=0}^{\infty}$ is a Markov chain in $\hat{S}$ on $(I \times \hat{\Omega}, \cal{I} \boxtimes \hat{\mathcal{F}}, \lambda \boxtimes \hat{P}^{\tilde{\omega}}),$ whose transition matrix $z^n$ at time $n-1$ is given by
\begin{align}
	z^n_{(k'J)(kl)}(\tilde \omega)&= \sum_{k_1, l_1,k' \in S}(1-\xi_{k_1 l_1}({\tilde{\omega}},n,\tilde{\tilde{p}}^{\tilde{\omega},n})) \sigma_{k_1 l_1}[k,l]\left({\tilde{\omega}},n,\tilde{\tilde{p}}^{\tilde{\omega},n}\right) {\theta}_{k_1l_1}({\tilde{\omega}},n,\tilde{p}^{\tilde{\omega},n}) \nonumber \\
	 &\quad \cdot {\eta}_{k' k_1}({\tilde{\omega}},n,\mathbb{E}^{\hat{P}^{\tilde{\omega}}}[\hat{p}^{n-1}]) \label{eq:TransitionMatrix1}\\ 
	z^n_{(k'l')(kl)}(\tilde \omega)&=\sum_{k_1, l_1,k',l' \in S} (1-\xi_{k_1 l_1}({\tilde{\omega}},n,\tilde{\tilde{p}}^{\tilde{\omega},n})) \sigma_{k_1 l_1}[k,l]\left({\tilde{\omega}},n,\tilde{\tilde{p}}^{\tilde{\omega},n}\right)\eta_{k' k_1}({\tilde{\omega}},n,\mathbb{E}^{\hat{P}^{\tilde{\omega}}}[\hat{p}^{n-1}]) \nonumber \\
	&\quad \cdot \eta_{l' l_1}({\tilde{\omega}},n,\mathbb{E}^{\hat{P}^{\tilde{\omega}}}[\hat{p}^{n-1}])  \label{eq:TransitionMatrix2} \\
	z^n_{(k'l')(kJ)}(\tilde \omega)&=\sum_{k_1,l_1 \in S} \xi_{k_1 l_1}({\tilde{\omega}},n,\tilde{\tilde{p}}^{\tilde{\omega},n}) \varsigma_{k_1 l_1}[k]\left({\tilde{\omega}},n,\tilde{\tilde{p}}^{\tilde{\omega},n}\right) \nonumber  \\
	& \quad \quad \cdot \eta_{k'k_1}({\tilde{\omega}},n,\mathbb{E}^{\hat{P}^{\tilde{\omega}}}[\hat{p}^{n-1}])\eta_{l'l_1}({\tilde{\omega}},n,\mathbb{E}^{\hat{P}^{\tilde{\omega}}}[\hat{p}^{n-1}])  \label{eq:TransitionMatrix3}\\
	z_{(k'J)(kJ)}^n(\tilde \omega)&={b}_{k}({\tilde{\omega}},n,\tilde{p}^{\tilde{\omega},n}) {\eta}_{k'k}({\tilde{\omega}},n,\mathbb{E}^{\hat{P}^{\tilde{\omega}}}[\hat{p}^{n-1}])\nonumber\\
	 & \quad +\sum_{k_1,l_1\in S} \xi_{k_1 l_1}({\tilde{\omega}},n,\tilde{\tilde{p}}^{\tilde{\omega},n}) \varsigma_{k_1 l_1}[k]\left({\tilde{\omega}},n,\tilde{\tilde{p}}^{\tilde{\omega},n}\right)\theta_{k_1l_1}({\tilde{\omega}},n,\tilde{p}^{\tilde{\omega},n})\nonumber \\
	 & \quad \quad  \cdot \eta_{k'k_1}({\tilde{\omega}},n,\mathbb{E}^{\hat{P}^{\tilde{\omega}}}[\hat{p}^{n-1}]).  \label{eq:TransitionMatrix4}
\end{align}
\item For $\lambda$-almost every $i$ and every $\lambda$-almost every $j$, the Markov chains $\lbrace \beta_i^n \rbrace_{n=0}^{\infty}$ and $\lbrace \beta_j^n \rbrace_{n=0}^{\infty}$ are independent on $( \hat{\Omega},  \hat{\mathcal{F}}, \hat{P}^{\tilde{\omega}})$.
\item For $\hat{P}^{\tilde{\omega}}$-almost every $\hat{\omega} \in \hat{\Omega}$, the cross sectional extended type process $\lbrace \beta^n_{\hat{\omega}} \rbrace_{n=0}^{\infty}$ is a Markov chain on $(I, \cal{I}, \lambda)$ with transition matrix $z^n$ at time $n-1$, which is defined in \eqref{eq:TransitionMatrix1}- \eqref{eq:TransitionMatrix4}.
\item {We have $\hat{P}^{\tilde{\omega}}$-a.s. that
\begin{align*}
	\mathbb{E}^{\hat{P}^{\tilde{\omega}}}[\check{p}^n_{kl}]=\check{p}^n_{kl} \quad \text{ and } \quad
	\mathbb{E}^{\hat{P}^{\tilde{\omega}}}[\check{\check{p}}^n_{kl}]=\check{\check{p}}^n_{kl} \quad \text{ and } \quad
	\mathbb{E}^{\hat{P}^{\tilde{\omega}}}[\hat{p}^n_{kl}]&=\hat{p}^n_{kl}. 
\end{align*}}
\end{enumerate}
\end{theorem}

\begin{proof}
	Fix $\tilde{\omega} \in \tilde{\Omega}$. Points 1. to 5. of {Theorem 3.14 in \cite{biagini_mazzon_oberpriller_meyer_brandis_2022}} follow directly by Lemma \ref{lemma:IndependenceRandomProcesses}, \ref{lemma:Auxiliary1} and \ref{lemma:Auxiliary2}. Moreover, Points 6. and 7. can be proven by using the same arguments as in the proof of Theorem 4 in \cite{RandomMatchingDiscrete}.
\end{proof}

\bibliography{Bibliography.bib}
\bibliographystyle{plainnat}
\end{document}